\documentclass{article}
\usepackage{arXiv}
\usepackage[utf8]{inputenc} 
\usepackage[T1]{fontenc}    
\usepackage{amsmath,amssymb,amsfonts,amsthm} 
\usepackage{newtxtext,newtxmath} 
\usepackage{mathrsfs} 
\usepackage{nicefrac}
\usepackage{bm} 
\usepackage{url}            
\usepackage{booktabs}       
\usepackage{graphicx} 
\usepackage{fancyhdr} 
\usepackage{float}
\usepackage{enumitem} 
\usepackage{subcaption} 
\usepackage{cite}
\newcommand{\upcite}[1]{\textsuperscript{\textsuperscript{\cite{#1}}}}
\usepackage{csquotes} 
\usepackage{lipsum} 
\usepackage{microtype}      
\usepackage{titlesec} 
\usepackage{geometry} 
\usepackage{hyperref} 

\newtheoremstyle{mystyle}{6pt}{6pt}{\itshape}{}{\bfseries}{.}{.5em}{}
\theoremstyle{mystyle}
\numberwithin{figure}{section}
\numberwithin{table}{section}
\numberwithin{equation}{section}
\newtheorem{theorem}{Theorem}[section]
\newtheorem{definition}[theorem]{Definition}
\newtheorem{lemmar}[theorem]{Lemma}
\newtheorem{corollary}[theorem]{Corollary}
\newtheorem{example}[theorem]{Example}

\newtheorem{remark}[theorem]{Remark}
\graphicspath{ {./images/} }
\title{Reducibility of Cartesian product quantum graph equipped with group action}
\author{
	Shimei Li \\
	School of Science\\
	Hebei University of Technology\\
	Tianjin 300401, China \\
	\texttt{202321102003@stu.hebut.edu.cn} \\
	\And
	Kai Zhang \\
	School of Science\\
	Hebei University of Technology\\
	Tianjin 300401, China \\
	\texttt{yizhipangkaiya@163.com} \\
	\And
	Jia Zhao \\
	School of Science\\
	Hebei University of Technology\\
	Tianjin 300401, China \\
	\texttt{zhaojia@hebut.edu.cn}\\
}
\begin{document}
	\maketitle
	\begin{abstract}
		We consider a Cartesian product quantum graph $\Gamma_{n_1}\Box\Gamma_{n_2}$ with standard vertex conditions, and complete the decomposition of Hilbert space $L^2(\Gamma_{n_1}\Box\Gamma_{n_2})$ and the Laplacian $\mathscr{H}$ on it by employing the relevant theories of group representation. The concept of $\Gamma_{n_1}\Box\Gamma_{n_2}$ equipped with the action of the cyclic group $G_{n_1}\times G_{n_2}$ is defined through the introduction of periodic quantum graph and cyclic groups. We also constructed its quotient graph and accomplish the decomposition of its secular determinant. Furthermore, under the condition that  $\gcd(n_1,n_2)=1$, it can be regarded as equivalent to a circulant graph $C_{n_1n_2}(n_1,n_2)$. This work also provides a new method for the construction of isospectral graphs. 
	\end{abstract}
	\section{Introduction}
	\par Quantum graph, as a simplified model, naturally finds applications in mathematics, physics, chemistry, and other engineering fields when one considers propagation of waves of various natures through a quasi-one-dimensional system. Such as the quantum wires\upcite{ref1}, carbon nano-structures\upcite{ref2}, photonic crystals\upcite{ref3} and others. The Cartesian product graph, by virtue of the specificity of its structure, can be used to describe the arrangement of atoms in certain crystals or study the evolution of quantum states as well as the implementation of quantum algorithms. The most common type is the periodic lattices. For example, Ondřej\upcite{ref4} investigated spectral gaps of the Hamiltonian on a periodic cuboidal (or generally hyperrectangular) lattice graphs with $\delta$ couplings in the vertices, established a connection between gap arrangements and continued fraction coefficients associated with lattice edge-length ratios, and thereby facilitated partial resolution of the inverse spectral problem. Shipman\upcite{ref5} constructs a class of non-symmetric periodic Schrödinger operator on a bilayer graph whose Fermi surface is reducible. The bilayer graph is formed by the Cartesian product of a $\mathbb{Z}^n$-periodic graph and an edge, and for AA-stacked bilayer graphene the Fermi surface is always reducible. Offering theoretical support to understand the electronic band structures and electron transport properties of materials like bilayer graphene deeply.
	\par The problem of decomposing function spaces on quantum graphs essentially aims to understand the mathematical properties and physical behaviors of quantum graphs from a structural perspective. The decomposition of these function spaces(e.g. $L^2(\Gamma)$, Sobolev space $H^s(\Gamma)$) helps us understand the core issues on quantum graphs, such as \enquote{state evolution}, \enquote{energy distribution}. Physically, functions in $L^2(\Gamma)$ correspond to the \enquote{wave functions} of quantum systems, edges of quantum graphs can model \enquote{quantum channels}, while vertices can simulate \enquote{scattering centers}. Decomposition of function spaces enables the quantification of transport behaviors of electrons or photons. 
	\par Carlson and others \upcite{ref6,ref7} studied the direct sum decomposition of the space of square integrable functions on regular metric trees; Jia Zhao\upcite{ref8} studied the spectrum of Schrödinger operators on regular metric trees that satisfy the $\delta$-conduction and the $\delta'$-condition at a point. By proving the large decomposition of the square-integrable function space, it was shown that the operator defined on the regular metric tree is unitarily equivalent to the operator on a line graph, and the necessary and sufficient condition for the spectrum of the Schrödinger operator on the graph to be purely discrete was obtained. 
	\par Group representation theory is a fundamental tool in various fields such as mathematics, statistics, and physics. Its applications in graph theory can be found in \cite{ref9}. Specifically, \cite{ref10} presents applications of group representation theory to calculating the eigenvalues of Cayley graphs. In the study of quantum graph, a key application is to use the symmetry of quantum graphs to facilitate the calculation of their spectra. Ben-Shach \upcite{ref11} and Parzanchevski \upcite{ref12} introduced the concept of quotient graphs for investigating isospectral quantum graphs. 
	\par The article is organized as follow. The basic knowledge of quantum graphs and group representation theory is briefly introduced in Section 2. The notion of Cartesian product quantum graph $\Gamma_{n_1}\Box\Gamma_{n_2}$ equipped with cyclic group $G_{n_1}\times G_{n_2}$ actions is defined in Section 3, through the introduction of periodic quantum graph and cyclic group. Main theorems are in Section 4: we completed the decomposition of the space $L^2(\Gamma_{n_1}\Box\Gamma_{n_2})$ and the Laplacian on this space by employing the relevant theories of group representation. And in Section 5, we also constructed its quotient graph and completed the decomposition of its secular determinant and give a spectial example that if the condition $\gcd(n_1,n_2)=1$ is satisfied, we can transformed the decomposition of the secular determinant of the Cartesian product quantum graph into the decomposition of the secular determinant of a circulant graph that is isomorphic to it.
	\section{Preliminaries}
	\label{sec:headings}
	\par It first presents the relevant definitions and theorems of quantum graphs\upcite{ref13} and group representation theory\upcite{ref14,ref15} involved in this article. 
	\subsection{Quantum graph}
	\par Let $\Gamma$ be a compact metric graph with a finite vertex set $V$ and edge set $E$, without loops and multiple edges. The numbers of vertices and edges in $\Gamma$ are denoted by $|V|$ and $|E|$ , respectively. In particular, a graph with $|V|=n$ vertices is denoted as $\Gamma_{n}$ and $V=\left\{v_i\right\}_{i=1}^n$. Each edge $e_j\in E,j=1,2,\dots,|E|$ in $\Gamma$ (the edge for adjacent $v_i,v_{\bar{i}}$ is denoted $v_i\sim v_{\bar{i}}$) is assigned a finite positive length $L_{e_j} \in(0, \infty]$.
	It means that each edge $e_j$ can be identified with an interval $[0,L_{e_j}]$ of the real line, and a coordinate $x_{e_j}$ can be assigned to each point along this interval.
	\par A function $f$ on the metric graph is denoted by the $n$-tuple
	$$
	f=\left\{f|_{e_1},f|_{e_2},\dots,f|_{e_{|E|}}\right\}
	$$ 
	of restrictions to the edges of $\Gamma$, in which each $f|_{e_j},j=1,2,\dots,|E|$ is a function of the interval $[0,L_{e_j}]$ that parameterizes the edge. That is, $f$ is a piecewise function on $\Gamma$. The Hilbert space $L^2(\Gamma)$ consists of functions that are measurable and square integrable on each $e_j$, and such that
	$$
	\lVert f\lVert^2_{L^2(\Gamma)}:=\sum_{e_j\in E,j=1}^{|E|}\lVert f\lVert^2_{L^2(e_j)}.
	$$
	\par Let $E_v$ denote the set of edges in $E$ that are incident to vertex $v\in V$, and $|E_v|$ denote the number of such edges. In a slight abuse of notation, we label these edges as $e_1,e_2,\dots,e_{|E_v|}$. At each vertex $v$ of $\Gamma$, the vertex conditions for $f\in L^2(\Gamma)$ can be expressed in the following form: 
	
	$$
	A_v
	\begin{pmatrix}
		f|_{e_1}(v)\\
		\vdots\\
		f|_{e_{|E_v|}}(v)
	\end{pmatrix}
	+B_v
	\begin{pmatrix}
		f'|_{e_1}(v)\\
		\vdots\\
		f'|_{e_{|E_v|}}(v)
	\end{pmatrix}
	=0,
	$$
	where $A_v$ and $B_v$ are $m\times d_v$ matrices and $m$ is any positive integer. The derivatives are assumed to be taken in the direction away from the vertex(i.e., into the edge), which will be called the outgoing directions. For example, the \textit{standard condition} at $v\in V$ is:
	\begin{equation}
		\left\{
		\begin{aligned}
			&f~\text{is continuous on vertex}~v: f|_{e_{j}}(v)=f|_{e_{\bar{j}}}(v), \forall e_{j},e_{\bar{j}}\in E_v,\\
			&\text{at vertex}~v~\text{one has}:                                                                                                                                                                                                                                                                                                                                                                                                                                                                                                                                                                                                                                                                                                                                                                                                                                                                                                                                                                                                                                                                                                                                                                                                                                                                                                                                                                                                                                                                                                                                                                                                                                                                                                                                                                                                                                                                                                                                                                                                                                                                                                                                                                                                                                                                                                                                                                                                                                                                                                                                                                                                                                                                                                              \varSigma_{e_j\in E_v}f'|_{e_j}(v)=0.
		\end{aligned}
		\right.
	\end{equation}
	Then the corresponding matrices $A_v$ and $B_v$ are
	$$
	A_v=
	\begin{pmatrix}
		1 & -1 & ~ & ~\\
		~ & \ddots & \ddots & ~\\
		~ & ~ & 1 & -1\\
		0 & \cdots & 0 & 0
	\end{pmatrix}
	,~B_v=
	\begin{pmatrix}
		0 & 0 & \cdots & 0 \\
		\vdots & \vdots & \ddots & \vdots \\
		0 & 0 & \cdots & 0\\
		1 & 1 & \cdots & 1
	\end{pmatrix}.
	$$
	In particular, when the boundary conditions (2.1) hold at a vertex of degree 2, the vertex can be eliminated, thus combining two adjacent edges into one smooth edge. Conversely, one usually also inserts degree-2 vertices satisfying the standard vertex conditions into edges for research purposes, and this process does not alter the spectrum of the quantum graph\upcite{ref12}.
	\par The Laplacian $\mathscr{H}$ takes the form $-d^2/dx^2$ on each edge. The domain of $\mathscr{H}$ consists of continuous functions on $\Gamma$ that, together with their derivatives along the edges, are square integrable, and satisfy the standard vertex condition at any vertex $v\in\Gamma$. This condition ensures that $\mathscr{H}$ is a self-adjoint operator in $L^2(\Gamma)$\upcite{ref13}. The pair $(\Gamma,\mathscr{H})$ is thus a quantum graph. A more detailed description will be provided in Section 3. 
	\subsection{Group representation theory}
	Let $G$ be a finite group, $\mathcal{K}$ be a field and $\mathcal{V}$ be a finite-dimensional vector space over $\mathcal{K}$. Denote by GL($\mathcal{V}$) the group of invertible linear transformations from $\mathcal{V}$ to itself. A group homomorphism $\rho: G\rightarrow \text{GL}(\mathcal{V})$ is called a \textit{linear $\mathcal{K}$-representation of $G$ in $\mathcal{V}$}(or just a \textit{representation of} $G$ for short). It means that $\rho$ satisfies:
	$$
	\begin{cases}
		\rho(g)\in\text{GL}(\mathcal{V}), &\forall g\in G;\\
		\rho(gh)=\rho(g)\rho(h), &\forall g,h\in G;\\
		\rho(e)=\text{1}_{\mathcal{V}}, 
	\end{cases}
	$$
	where $e$ is the identity element of $G$, and $\text{1}_{\mathcal{V}}$ is the identity transformation on $\mathcal{V}$. The dimension of $\mathcal{V}$ is called the \textit{dimension} of $\rho$, denoted by dim($\rho$). 
	\par Any time a natural representation can be expressed (up to isomorphism) as a direct sum or an extension of smaller representation, but it not always possible because a representation $\rho$ may have no non-trivial \textit{subrepresentation}\upcite{ref14} to try to \enquote{peel off}. This leads to the following special case:
	\begin{definition}
		\rm A $\mathcal{K}-$representation $\rho$ of $G$ acting on $\mathcal{V}$ is \textit{irreducible} if and only if $\mathcal{V}\neq$ 0 and there is no subspace $\mathcal{W}\subset\mathcal{V}$ stable under $\rho$~(i.e., $\rho(g)(\mathcal{W})\subset\mathcal{W}$ for all $g\in\mathcal{V}$), except 0 and $\mathcal{V}$ itself. 
	\end{definition}
	\par Finite groups often arise in the study of symmetrical objects, particularly when those objects admit only a finite number of structure-preserving transformations. For example, cyclic groups describe objects that possess only rotational symmetry:
	\begin{definition}
		\rm A group $G$ is called a \textit{cyclic group} if there is an element $g\in G$ such that every element of $G$ is some integral power of $g$. The group $G$ is said to be \textit{generated by $g$} and $g$ is called a \textit{generator} of $G$, then $G$ is denoted by $G=\langle g \rangle=\left\{g^n:n\in \mathbb{Z}\right\}$.
	\end{definition}
	\par Let $G$ be a finite cyclic group of order $n$, generated by $g$. We denote $G$ as $G_n=\left\{g,g^2,g^3,\dots,g^n=e\right\}$, where $e$ is the identity element of $G_n$. Let $\mathcal{K}=\mathbb{C}$, all finite-dimensional irreducible complex representations of a cyclic group $G_n$ are one-dimensional, and there are exactly $n$ such representations\upcite{ref14}. For instance, the irreducible complex representations of $G_n$ are summarized in Table 2.1, where $\omega=\text{e}^{\frac{2\pi\text{i}}{n}}$ denotes an $n-$th primitive root of unity.
	\begin{table}[H]
		\caption{Irreducible complex representations of $G_n$.}
		\centering
		\begin{tabular}{ccccccc}
			\toprule
			$~$ & $e$ & $g$ & $g^2$ & $g^3$ & $\cdots$ & $g^{n-1}$\\
			\midrule
			$\rho_0$ & $1$ & $1$ & $1$ & $1$ & $\cdots$ & $1$\\
			$\rho_1$ & $1$ & $\omega$ & $\omega^2$ & $\omega^3$ & $\cdots$ & $\omega^{n-1}$\\
			$\rho_2$ & $1$ & $\omega^2$ & $\omega^4$ & $\omega^6$ & $\cdots$ & $\omega^{2(n-1)}$\\
			$\vdots$ & $\vdots$ & $\vdots$& $\vdots$ & $\vdots$ & $\vdots$ & $\vdots$\\
			$\rho_{n-1}$ & $1$ & $\omega^{n-1}$ & $\omega^{2(n-1)}$ & $\omega^{3(n-1)}$ & $\cdots$ & $\omega^{(n-1)^2}$\\
			\bottomrule
		\end{tabular}
	\end{table}
	
	\par Moreover, for the case $G_{n_1}\otimes G_{n_2}$. Let $\mathcal{K}=\mathbb{C}$, cyclic group $G_{n_1}=\langle g_1 \rangle$ and $G_{n_2}=\langle g_2 \rangle$ have $n_1$ and $n_2$ one-dimensional irreducible complex representations respectively, denoted as $\rho_s$ and $\varrho_t$, $s=0,1,\dots,n_1-1, t=0,1,\dots,n_2-1$. Then, the external tensor product $\rho_s\boxtimes\varrho_t$ is an irreducible representation of  $G_{n_1}\otimes G_{n_2}$~(\cite{ref14}, Proposition 2.3.23), denoted as $\tau_{s,t}$. That is, for any $(g_1^\kappa,g_2^\iota)\in G_{n_1}\otimes G_{n_2}, \kappa=1,\dots,n_1,\iota=1,\dots,n_2$:
	$$
	\tau_{s,t}(g_1^\kappa,g_2^\iota)=\rho_s(g_1^\kappa)\boxtimes\varrho_t(g_2^\iota), s=0,1,\dots,n_1-1;t=0,1,\dots,n_2-1.
	$$
	And if $\omega_1$ is the $n_1$-th unit root, $\omega_2$ is the $n_2$-th unit root, then $\tau_{s,t}(g_1^\kappa,g_2^\iota)=(\omega_1^s)^\kappa(\omega_2^t)^\iota$. Therefore, the group $G_{n_1}\otimes G_{n_2}$ has $n_1n_2$ irreducible complex representations, and for any element $(g_1^\kappa,g_2^\iota)\in G_{n_1}\oplus G_{n_2}$, the sum of these $n_1n_2$ irreducible complex representations takes the following form:
	\begin{equation}
		\sum_{s=1}^{n_1}\sum_{t=1}^{n_2}\tau_{s,t}(g_1^\kappa,g_2^\iota)=
		\begin{cases}
			n_1n_2, &\text{if}~\kappa=n_1, \iota=n_2,\\
			0, &\text{the other cases}.   	
		\end{cases}
	\end{equation}
	\section{Construction of Cartesian product graph with group action}
	\par Similar to the construction of Cartesian product graphs in graph theory, defining a Cartesian product metric graph funther requires the specifying the edge lengths of the graphs. We present the formal  definition as follows:
	\begin{definition}
		\rm Let $\Gamma_{n_1}$ and $\Gamma_{n_2}$ are metric graphs with $n_1,n_2$ vertices respectively, $V(\Gamma_{n_1})=\left\lbrace u_i \right\rbrace_{i=1}^{n_1}, V(\Gamma_{n_2})=\left\lbrace v_i \right\rbrace_{i=1}^{n_2}$. The \textit{Cartesian product metric graph} of them is defined as $\Gamma_{n_1} \Box \Gamma_{n_2}$. Its vertex set is $V(\Gamma_{n_1})\times V(\Gamma_{n_2})$ and edge set consists of all pairs $(u_{i},v_{i})~(u_{\bar{i}},v_{\bar{i}})$~such that either $u_{i}\sim u_{\bar{i}}\in E(\Gamma_{n_1})$ and $v_{i}=v_{\bar{i}}$, or $v_{i}\sim v_{\bar{i}}\in E(\Gamma_{n_2})$ and $u_{i}=u_{\bar{i}}$. Two vertices $(u_{i},v_{i})$ and $(u_{\bar{i}},v_{\bar{i}})$ are adjacent if and only if one of the following cases holds: 
		\begin{itemize}
			\item[1)] $u_i=u_{\bar{i}}$, $v_i$ and $v_{\bar{i}}$ are adjacent in graph $\Gamma_{n_2}$. In this case, the length between them is the length of edge between $v_i$ and $v_{\bar{i}}$ in $\Gamma_{n_2}$;
			\item[2)] $v_i=v_{\bar{i}}$, $u_i$ and $u_{\bar{i}}$ are adjacent in graph $\Gamma_{n_1}$. In this case, the length between them is the length of edge between $u_i$ and $u_{\bar{i}}$ in $\Gamma_{n_1}$. 
		\end{itemize}
	\end{definition}
	\par For example, Figure 3.1 shows a metric graph $\Gamma_2$ with edge length 1, a metric graph $\Gamma_3$ with two edges of length 1 and 2, and their Cartesian product graph $\Gamma_2\Box\Gamma_3$.
	\begin{figure}[H]
		\centering 
		\includegraphics[width=10cm]{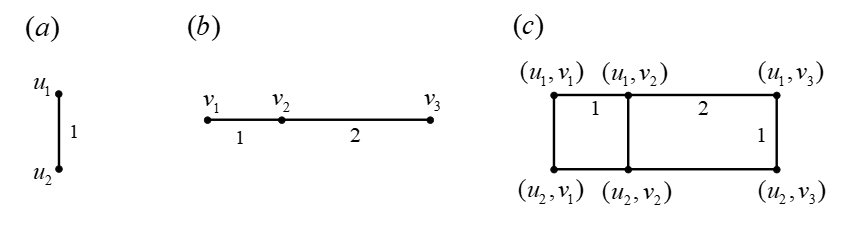}\\
		\caption{(a)Metric graph $\Gamma_2$ with edge length 1. (b)Metric graph $\Gamma_3$ with edges length 1 and 2. (c)Cartesian product graph $\Gamma_2\Box\Gamma_3$.}
	\end{figure} 
	Analogous to the structure of $\mathbb{Z}^n$-periodic quantum graphs\upcite{ref13}, quantum graphs $\Gamma$ with group $G_n$-action can also be constructed and described. First, such a graph must include an underlying graph $\Gamma/G_n$ with finite vertex set $V=V(\Gamma/G_n)$ and edge set $E=E(\Gamma/G_n)$, endowed with an action of the group $G_n$ that preserves vertex-edge incidences and for which $\Gamma/G_n$ is a finite graph. The action of $g^\kappa\in G_n$ on a vertex or edge of $\Gamma$ is denoted by $v\mapsto g^\kappa v$ or $e_j\mapsto g^\kappa e_j$ respectively. A fundamental domain of the $G_n$-action is denoted by $W$, which is assumed to contain finitely many vertices and edges. 
	\par Next, building on the Definition 3.1, we assume this metric graph is invariant under the action of $G_n$, and let the action of $g^\kappa\in G_n$ on a point $y$ in $\Gamma$~($y$ may be in the interior of an edge or at an endpoint corresponding to a vertex) be denoted by $y\mapsto g^\kappa y$. This enables us to define standard function spaces on any edges in $\Gamma$, such as $H^2(e_j)$. Then we can give the definition of a graph equipped with a cyclic group action:
	\begin{definition}
		\rm Let $G_n$ be a cyclic group. A metric graph $\Gamma$ is said to be a \textit{graph equipped with $G_n$ action} if the mapping $(g^\kappa,y)\in G_n\times \Gamma\mapsto g^\kappa y\in \Gamma$ satisfies:
		\begin{enumerate}
			\item \textit{Group Action:} 
			$\forall g^\kappa\in G_n,$ the mapping $y\mapsto g^\kappa y$ is a bijection of~$\Gamma$; $\forall y\in \Gamma, ey=y,$ where $e\in G_n$ is the identity element; $\forall g^\kappa,g^{\bar{\kappa}}\in G_n, y\in \Gamma, (g^\kappa g^{\bar{\kappa}})y=g^\kappa(g^{\bar{\kappa}}y)$.
			\item \textit{Continuity:} $\forall g^\kappa\in G_n$, the mapping $y\mapsto g^\kappa y$ from $\Gamma$ to itself is continuous. 
			\item \textit{Faithfulness:} If $y\in\Gamma$ and $g^\kappa y=y$, then $g^\kappa=e$.
			\item \textit{Discreteness:} For any $y\in \Gamma$, there is a neighborhood $U$ of $y$ such that $g^\kappa y\notin U$ for $g^\kappa\neq e$.
			\item \textit{Co-compactness}: The space of orbits $\Gamma/G_n$ is compact, i.e. the entire graph can be obtained by the $G_n$-shifts of a compact subset. 
			\item \textit{Structure preservation:} 
			\begin{itemize}
				\item For vertices $v_i$ and $v_{\bar{i}}$, $g^\kappa v_i\sim g^\kappa v_{\bar{i}}$ if and only if $v_i\sim v_{\bar{i}}$~(\enquote{$\sim$} means adjacent in there), and $G_n$ acts bijectively on the set of edges.
				\item In the case of a metric or quantum graph, the action preserves the length of edges: $L_{g^\kappa e_j}=L_{e_j}$.
				\item In the case of a quantum graph, the action commutes with the Hamiltonian $\mathscr{H}$ (and in particular, preserves the vertex conditions).
			\end{itemize}
		\end{enumerate}
	\end{definition}
	\par For example, a graph on $n$ vertices is said to be a \textit{circulant graph} $C_n({\rm\textbf{s}})$ if it admits an action by the cyclic group $G_n$ and is defined by a vector ${\rm \textbf{s}}=(s_1,s_2,s_3,...,s_k)$. Here, each component $s_{\bar{k}}\in \mathbb{N}^*$~(for $\bar{k}=1,2,...,$)~satisfies $1\leqslant s_{\bar{k}} \leqslant n/2$, and two vertices $v_i$ and $v_{\bar{i}}$ are adjacent if and only if $i-\bar{i}\equiv \pm s_{\bar{k}}$~(mod~$n$)~for some $\bar{k}\in{1,2,...,k}$. 
	\par The fundamental domain $W \subset \Gamma$ is a compact (or finite in the discrete case) subset for the action of $G_n$ on $\Gamma$, if it both satisfies the following conditions: 
	\begin{enumerate}
		\item The union of all $G_n-$ shifts of $W$ covers the whole $\Gamma$, i.e., $\bigcup_{g^\kappa\in G_n}g^\kappa W=\Gamma$.
		\item Different shifted copies of $W$, i.e., $g^\kappa W$ and $g^{\bar{\kappa}}W$ with $g^{\kappa}\neq g^{\bar{\kappa}}\in G_n$, intersect in at most finitely many points, none of which are vertices. 
	\end{enumerate}
	It is obvious that the choice of the fundamental domain $W$ is not unique.
	For example, a fundamental domain of $C_6\left(1,2\right)$ is shown in Figure 3.2(b). We introduce virtual vertices at the center of edges in $C_6\left(1,2\right)$. This process effectively doubles the number of edges in $C_6\left(1,2\right)$, with the new edges represented by dashed lines. The edges incident to the original vertices $v_i$ of $C_6\left(1,2\right)$ are labeled as $e_{i,1},e_{i,2},...,e_{i,d_{v_i}}$. Each original vertex $v_i$ together with all adjacent virtual vertices and the associated edges forms a fundamental domain.
	\begin{figure}[H]
		\centering 
		\includegraphics[width=8cm]{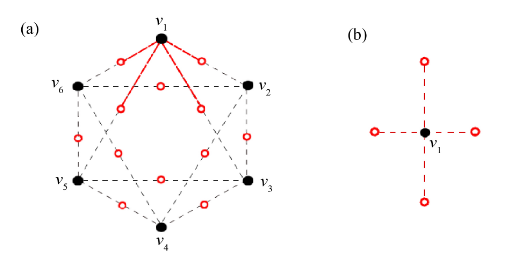}\\
		\caption{(a)The circulant graph $C_6\left(1,2\right)$ with dummy vertices. (b)The fundamental domain of $C_6\left(1,2\right)$.}
		\label{fig:fig2}
	\end{figure} 
	\par Last, we renders $\Gamma$ a quantum graph with $G_n$ actions by pairing it with Laplacian $\mathscr{H}:=-d^2/dx^2$ after parameterizing each edge $e_j\in E(\Gamma)$, it can also commute with the $G_n$ action. Let functions $f=\left\{f|_{e_j}\right\}_{e_j\in E(\Gamma)}$ defined on $\Gamma$ satisfy the condition (2.1) at each vertex, $f'|_{e_j}(v)$ is the derivative of $f|_{e_j}$ at the vertex $v$ and has outgoing direction. Let 
	$$
	H^2(\Gamma)=\left\{f=\left\{f|_{e_j}\right\}_{e_j\in E(\Gamma)}:f \text{~is continuous};f|_{e_j}\in H^2(e_j),\forall e_j\in E(\Gamma);f,f',f''\in L^2(\Gamma)\right\},
	$$
	in which $f'$ and $f''=-\frac{d^2}{dx^2}$ are taken on each edge with respect to the coordinates introduced above. Then the domain of $\mathscr{H}$ and its action form is:
	\begin{equation}
		\mathcal{D}(\mathscr{H})=\left\{f\in H^2(\Gamma):f~\text{satisfies (2.1) forall}~v\in V(\Gamma)\right\},
	\end{equation}
	\begin{equation}
		(\mathscr{H}f)(x)=-f''(x).
	\end{equation}
	The standard condition (2.1) makes sense because $f\in H^2(\Gamma)$ has well defined derivatives at the endpoints of each edge and $\mathscr{H}$ is self-adjoint in $L^2(\Gamma)$ \upcite{ref13}. Therefore, let $\Gamma_{n_1}$ and $\Gamma_{n_2}, n_1,n_2\in\mathbb{N}^*$ be the metric graph equipped with group $G_{n_1}$ and $G_{n_2}$ action respectively. The Cartesian product graph $\Gamma_{n_1}\Box\Gamma_{n_2}$ will be equipped with the group $G_{n_1}\otimes G_{n_2}$ action and the cartesian product quantum graph $\Gamma_{n_1}\Box\Gamma_{n_2}$ equipped with group $G_{n_1}\otimes G_{n_2}$ action is a triple 
	$$
	\begin{Bmatrix}
		\rm \Gamma_{n_1}\Box\Gamma_{n_2}, \text{Laplacian}~\mathscr{H},· \text{standard ~conditions(2.1)}
	\end{Bmatrix}.
	$$
	\section{The decomposition of function space and Laplacian}
	\par In this secction, we mainly give the decomposition of the space  $L^2(\Gamma_{n_1}\Box\Gamma_{n_2})$ and the Laplacian defined on the metric graph $\Gamma_{n_1}\Box\Gamma_{n_2}$.
	\begin{theorem}	
		\rm Let $\Gamma_{n_1}\Box\Gamma_{n_2}$ is a Cartesian product graph equipped with the group	$G_{n_1}\otimes G_{n_2}$ action. Let $V(\Gamma_{n_1}\Box\Gamma_{n_2})=\left\{v_i\right\}_{i=1}^{n_1n_2}$, the edges connected to each vertex $v_i$ are denoted as $e_{i,1}, e_{i,2},\dots,e_{i,|E_{v_i}|}$. Then the space~$L^2\left(\Gamma_{n_1}\Box\Gamma_{n_2}\right)$ can be decomposed into the direct sum of $n_1n_2$ square-integrable function spaces, i.e.
		$$
		L^2\left(\Gamma_{n_1}\Box\Gamma_{n_2}\right)\cong \oplus_{s=0}^{n_1-1}\left(\oplus_{t=0}^{n_2-1}\mathcal{F}_{s,t}\right),
		$$
		where~$\mathcal{F}_{s,t}$ is square-integrable function spaces on $\Gamma_{n_1}\Box\Gamma_{n_2}$. 
	\end{theorem}
	\begin{proof}
		\par Introduce a dummy vertex at the midpoint of each edge of $\Gamma_{n_1}\Box\Gamma_{n_2}$. The vertex $v_1$, together with all its adjacent dummy vertices and the edges incident to $v_1$ forms a fundamental domain $W$. Take any function $f\in L^2(\Gamma_{n_1}\Box \Gamma_{n_2})$ restricted to the edge $e_{i,m}$ is denoted by $f|_{e_{i,m}}, i=1,2,...,d_{v_i}$. Define a mapping $P$ from $L^2(\Gamma_{n_1}\Box\Gamma_{n_2})$ to $\oplus_{s=0}^{n_1-1}\left(\oplus_{t=0}^{n_2-1}\mathcal{F}_{s,t}\right)$. According to the irreducible representation of $G_{n_1}\otimes G_{n_2}$, there is a mapping $P_{s,t}$ from the function $f$ to the function $f_{s,t}\in\mathcal{F}_{s,t}$, for $s=0,1,\dots n_1-1,t=0,1,\dots,n_2-1$. Let $g_{\kappa,\iota}=\left(g_1^\kappa,g_2^\iota\right),\kappa=1,\dots,n_1,\iota=1,\dots,n_2$, then  
		\begin{equation}
			\begin{aligned}
				P_{s,t}f\mid_{e_{i,m}}
				&=f_{s,t}\mid_{e_{i,m}}\\
				&=\tfrac{1}{n_1n_2}\left[\tau_{s,t}(g_{1,1})f\mid_{g_{1,1}e_{i,m}}+\tau_{s,t}(g_{2,1})f\mid_{g_{2,1}e_{i,m}}+\cdots+\tau_{s,t}(g_{n_1,1})f\mid_{g_{n_1,1}e_{i,m}}\right.\\
				&\quad\quad+\tau_{s,t}(g_{1,2})f\mid_{g_{1,2}e_{i,m}}+\tau_{s,t}(g_{2,2})f\mid_{g_{2,2}e_{i,m}}+\cdots+\tau_{s,t}(g_{n_1,2})f\mid_{g_{n_1,2}e_{i,m}}\\
				&\quad\quad+\dots\\
				&\quad\quad\left.+\tau_{s,t}(g_{1,n_2})f\mid_{g_{1,n_2}e_{i,m}}+\tau_{s,t}(g_{2,n_2})f\mid_{g_{2,n_2}e_{i,m}}+\cdots+\tau_{s,t}(g_{n_1,n_2})f\mid_{g_{n_1,n_2}e_{i,m}}\right].
			\end{aligned}
		\end{equation}
		where $g_{n_1,n_2}\in G_{n_1}\otimes G_{n_2}$ is the identity element. Then the mapping $P$ is surjective. Since
		\begin{equation}
			\sum_{s=0}^{n_1-1}\sum_{t=0}^{n_2-1}\tau_{s,t}(g_{\kappa,\iota})=
			\begin{cases}
				n_1n_2, &\text{if}~\kappa=n_1, \iota=n_2,\\
				0, &\text{the other cases}.   	
			\end{cases}
		\end{equation}
		then 
		\begin{equation}
			f=f_{1,1}+f_{2,1}+\cdots+f_{n_1,1}+f_{1,2}+f_{2,2}+\cdots+f_{n_1,2}+\dots+f_{1,n_2}+f_{2,n_2}+\cdots+f_{n_1,n_2}.
		\end{equation}
		and
		\begin{equation}
			\big\| f\big\|_{L^2(\Gamma_{n_1}\Box\Gamma_{n_2})}^2=\sum_{s=0}^{n_1-1}\sum_{t=0}^{n_2-1}\big\| f^2_{s,t}\big\|_{\mathcal{F}_{s,t}},
		\end{equation}
		it is proved that the mapping $P$ is an isometric isomorphic mapping. 
	\end{proof}
	\begin{corollary}
		\rm If $n_1$ and $n_2$ are coprime, it can be proved that $G_{n_1}\otimes G_{n_2}\cong G_{n_1n_2}=\langle g \rangle$(\cite{ref15}, Theorem 9.3). That is, there exists an isomorphism mapping the group element $\left(g_1^\kappa,g_2^\iota\right)\in G_{n_1}\otimes G_{n_2}$ to $g^\epsilon,\epsilon=\kappa n_2+\iota n_1\mod (n_1n_2)$. The element $g^\epsilon$ of $G_{n_1n_2}$ satisfies $g^\epsilon e_{1,m}=e_{\epsilon+1,m}$ for $\epsilon=0,1,...,n_1n_2-1$, where $g^0=g^{n_1n_2}=e$. 
		All irreducible complex representations of $G_{n_1n_2}$ are denoted as $\rho_{r}, r=0,1,\dots,n_1n_2-1$, then $\tau_{s,t}(g_1^\kappa,g_2^\iota)=\rho_{r}(g^\epsilon), r=(sn_2+tn_1)\mod (n_1n_2)$.
		And the space decomposition can be rewritten as $L^2(\Gamma_{n_1}\Box\Gamma_{n_2})\cong\oplus_{r=0}^{n_1n_2-1}\mathcal{F}_r$.
		The function $f_r\in \mathcal{F}_r$ can be represented by its restriction $f_r|_{W}$ on the fundamental domain $W$ due to (4.1):
		$$
		f_r|_{g^\epsilon W}=
		\begin{cases}
			f_r|_W, &\epsilon=0,\\
			\rho_r(g^{\epsilon-1})f_r|_W, &\epsilon=1,2,3,...,n_1n_2-1.
		\end{cases}
		$$
	\end{corollary}
	\par In addition, circulant graph $C_{n_1n_2}(n_1,n_2)$ also admits the action of the cyclic group $G_{n_1n_2}$. Due to the Theorem 3.1 in \cite{ref16}, if $n_1,n_2$ are coprime, i.e. $gcd(n_1,n_2)=1$, then $C_{n_1n_2}(n_1,n_2)\cong\Gamma_{n_1}\Box\Gamma_{n_2}$. Let $g_1$ denote a clockwise rotation by $2\pi/n_1$ degrees and $g_2$ denote a counterclockwise rotation by $2\pi/n_2$ degrees. Then the vertex $(g_1^\kappa,g_2^\iota)v_1\in\Gamma_{n_1}\Box \Gamma_{n_2}$ corresponds to vertex $v_{\epsilon+1}\in C_{n_1n_2}\left(n_1,n_2\right), \epsilon=\kappa n_2+\iota n_1\mod (n_1n_2)$. Therefore, we can reduce the decomposition of the function space on $\Gamma_{n_1}\Box\Gamma_{n_2}$ to that on the $C_{n_1n_2}(n_1,n_2)$ by Corollary 4.2.
	\begin{example}
		\rm Let $n_1=3,n_2=4$, then $gcd(3,4)=1$ and $G_3\otimes G_4\cong G_{12}, \Gamma_3\Box\Gamma_4\cong C_{12}(3,4)$. so we can transform the decomposition of the function space on $\Gamma_3\Box\Gamma_4$ into the $C_{12}\left(3,4\right)$.  
		\begin{figure}[H]
			\centering 
			\includegraphics[width=5cm]{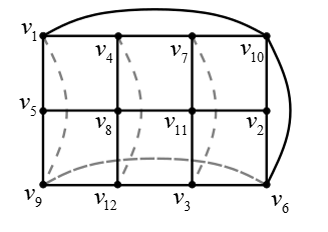}\\
			\caption{The cartesian graph of the $\Gamma_3$ and $\Gamma_4$.}
			\label{fig:fig3}
		\end{figure}
		\begin{figure}[H]
			\centering 
			\includegraphics[width=10cm]{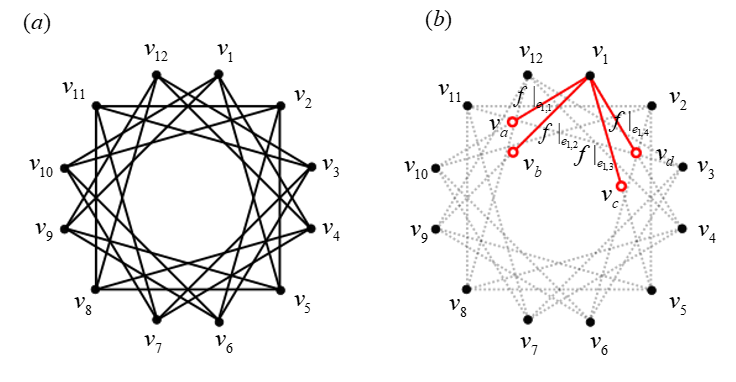}\\
			\caption{(a)~The circulant graph $C_{12}(3,4)$.~(b)~The fundamental domain of $C_{12}(3,4)$.}
			\label{fig:fig4}
		\end{figure}
		The circulant graph $C_{12}(3,4)$ is equipped with the cyclic group $G_{12}$ action. The fundanmental domain $W$ of $C_{12}(3,4)$ is constructed as shown in Figure 4.2(b). The 12 irreducible representations of $G_{12}$ yield the identification: 
		$$
		L^2(C_{12}(3,4))=\oplus_{l=0}^{11}\mathcal{F}_r.
		$$
		For any $f\in L^2(C_{12}(3,4))$, the $r$-th irreducible representation of the cyclic group $G_{12}$ yields $f_r\in\mathcal{F}_r$, for $r=0,1,2,\dots,11$, 
		\begin{equation}
			\begin{aligned}
				&f_r|_{e_{1,1}}=\frac{\rho_r(e)f|_{e_{1,1}}+\rho_r(g)f|_{e_{2,1}}+\rho_r(g^2)f|_{e_{3,1}}+\cdots+\rho_r(g^{11})f|_{e_{12,1}}}{12},\\
				&f_r|_{e_{1,2}}=\frac{\rho_r(e)f|_{e_{1,2}}+\rho_r(g)f|_{e_{2,2}}+\rho_r(g^2)f|_{e_{3,2}}+\cdots+\rho_r(g^{11})f|_{e_{12,2}}}{12},\\
				&\vdots\\
				&f_r|_{e_{12,4}}=\frac{\rho_r(e)f|_{e_{12,4}}+\rho_r(g)f|_{e_{1,4}}+\rho_r(g^2)f|_{e_{2,4}}+\cdots+\rho_r(g^{11})f|_{e_{11,4}}}{12},\\
			\end{aligned}
		\end{equation}
		Moreover, the function $f_r\in\mathcal{F}_r$ can be expressed in terms of its restriction to the fundamental domain $W$, denoted by $f_r|_W$:
		$$
		f_r|_{g^\epsilon W}=
		\begin{cases}
			f_r|_W &, \epsilon=0,\\
			\rho_r(g^{\epsilon-1})f_r|_W &, \epsilon=1,2,\dots,11.
		\end{cases}
		$$
		This is consistent with the conclusion in Remark 4.2.
	\end{example}
	\begin{theorem}
		\rm The domain of Laplacian $\mathscr{H}:=-\frac{d^2}{dx^2}$ defined on $\Gamma_{n_1}\Box\Gamma_{n_2}$ is 
		\begin{equation}
			\begin{aligned}
				\mathcal{D}(\mathscr{H})=\{f\mid 
				& f, f^{\prime}, f^{\prime \prime} \in L^2\left(\Gamma_{n_1} \square \Gamma_{n_2}\right),\\ 
				& f~\text{\rm satisfies the condition (2.1) at the vertices of~} \Gamma_{n_1} \square \Gamma_{n_2}\}.
			\end{aligned}
			\nonumber
		\end{equation}
		then $\mathscr{H}$ is unitarily equivalent to the direct sum of the Laplacian $\mathscr{H}_{s,t}$ defined on the space $\mathscr{\mathcal {F}}_{s,t}, s=0,1,\dots,n_1-1, t=0,1,\dots,n_2-1$, i.e. 
		$$
		\mathscr{H} \cong\oplus_{s=0}^{n_1-1}\left( \oplus_{t=0}^{n_2-1}\mathscr{H}_{s,t}\right).
		$$
	\end{theorem}
	\begin{remark}
	\end{remark}
	\begin{itemize}
		\item The domain of $\mathscr{H}_{s,t}$ is 
		\begin{equation}
			\begin{aligned}
				\mathcal{D}(\mathscr{H}_{s,t})=\{f_{s,t} \mid 
				& f_{s,t}, f_{s,t}^{\prime}, f_{s,t}^{\prime \prime} \in \mathcal{F}_{s,t},\\ 
				& f_{s,t}~\text{satisfies the condition (2.1) at the vertices of}~\Gamma_{n_1}\Box\Gamma_{n_2},\\
				& f_{s,t}~\text{satisfies} 
				\left\{
				\begin{aligned}			
					&f_{s,t}|_{e_{i,m}}\left(\tilde{v}\right)=\tau_{s,t}\left(g_{\kappa,\iota}\right)f_{s,t}|_{e_{i,m}}\left(\tilde{v}\right),\\
					&f_{s,t}^\prime|_{e_{i,m}}\left(\tilde{v}\right)+\tau_{s,t}\left(g_{\kappa,\iota}\right)f_{s,t}^\prime|_{e_{i,m}}\left(\tilde{v}\right)=0,
				\end{aligned}
				\right.
				\\
				&\text{at the dummy vertex}~\tilde{v}~\text{connecting edges}~ e_{i,m}~\text{and}~e_{\bar{i},m}=g_{\kappa,\iota}e_{i,m}
				\}.
			\end{aligned}
		\end{equation}
		\item The following spectral relation can be obtained from the decomposition of the operators, 
		$$
		\sigma(\mathscr{H})=\bigcup^{n_1-1}_{s=0}\left(\bigcup^{n_2-1}_{t=0}\sigma(\mathscr{H}_{s,t})\right).
		$$
	\end{itemize}
	\begin{proof}
		\par Since 
		$$
		\mathcal{D}(\mathscr{H})\subset L^2(\Gamma_{n_1}\Box \Gamma_{n_2}),
		$$
		then a direct sum decomposition of the domain of the operator follows from Theorem 4.1, i.e. 
		$$
		\mathcal{D}(\mathscr{H})\cong\oplus_{s=0}^{n_1-1}\left(\oplus_{t=0}^{n_2-1}\mathcal{D}(\mathscr{H}_{s,t})\right).
		$$
		For any $f_{s,t}|_{e_{i,m}}\in \mathcal{F}_{s,t}$,
		$$
		\mathscr{H}_{s,t}f_{s,t}|_{e_{i,m}}=-f^{\prime\prime}_{s,t}|_{e_{i,m}}\in\mathcal{F}_{s,t},
		$$
		therefore
		$$
		\mathscr{H}_{s,t}(\mathcal{F}_{s,t})\subset \mathcal{F}_{s,t}.
		$$
		To sum up, we can get 
		$$
		\mathscr{H}\cong\oplus_{s=0}^{n_1-1}\left(\oplus_{t=0}^{n_2-1}\mathscr{H}_{s,t}\right).
		$$
		In the following, it is proved that $f_{s,t}$ in $\mathcal{D}(\mathscr{H}_t)$ satisfies the condition(4.6). Let $v_i$ is the original vertex of $\Gamma_{n_1}\Box\Gamma_{n_2}$, doniting 
		$$
		\langle f_{s,t}|_{e_{i,m}}\rangle_{\kappa,\iota}=\frac{\tau_{s,t}(g_{\kappa,\iota})f|_{g_{\kappa,\iota}e_{i,m}}}{n_1n_2},
		$$ 
		$f$ satisfies the standard condition (2.1) at $v_i$ , therefore 
		\begin{equation}
			\left\{
			\begin{aligned}
				&\langle f_{s,t}|_{e_{i,1}}\rangle_{\kappa,\iota}(v_i)=\langle f_{s,t}|_{e_{i,2}}\rangle_{\kappa,\iota}(v_i)=\cdots=\langle f_{s,t}|_{e_{i,d_{v_i}}}\rangle_{\kappa,\iota}(v_i),\\
				&\langle f^\prime_{s,t}|_{e_{i,1}}\rangle_{\kappa,\iota}(v_i)+\langle f^\prime_{s,t}|_{e_{i,2}}\rangle_{\kappa,\iota}(v_i)+\cdots+\langle f^\prime_{s,t}|_{e_{i,d_{v_i}}}\rangle_{\kappa,\iota}(v_i)=0,\\
			\end{aligned}
			\right.
		\end{equation}
		so we get 
		\begin{equation}
			\left\{
			\begin{aligned}
				&f_{s,t}|_{e_{i,1}}(v_i)=f_{s,t}|_{e_{i,2}}(v_i)=\cdots=f_{s,t}|_{e_{i,d_{v_i}}}(v_i),\\
				&f^\prime_{s,t}|_{e_{i,1}}(v_i)+f^\prime_{s,t}|_{e_{i,2}}(v_i)+\cdots+f^\prime_{s,t}|_{e_{i,d_{v_i}}}(v_i)=0,\\
			\end{aligned}
			\right.
		\end{equation}
		thus, $f_{s,t}$ satisfies the standard condition (2.1) at any of the original vertices of $\Gamma_{n_1}\Box\Gamma_{n_2}$. Since $e_{\bar{i},m}=g_{\kappa,\iota}e_{i,m}$, at the dummy vertex $\tilde{v}$ connecting edges $e_{i,m}$ and $e_{\bar{i},m}$, $f_{s,t}$ is also satisfied 
		\begin{equation}
			\left\{
			\begin{aligned}
				&f_{s,t}|_{e_{i,m}}(\tilde{v})=\tau_{s,t}(g_{\kappa,\iota})f_{s,t}|_{e_{\bar{i},m}}(\tilde{v}),\\
				&f^\prime_{s,t}|_{e_{i,m}}(\tilde{v})+\tau_{s,t}(g_{\kappa,\iota})f^\prime_{s,t}|_{e_{\bar{i},m}}(\tilde{v})=0.
			\end{aligned}
			\right.
		\end{equation} 
		\par This concludes the proof.
	\end{proof}
	\section{Secular determinant and its decomposition}
	In this section, we decompose the secular determinant of the Cartesian product graph $\Gamma_{n_1}\Box\Gamma_{n_2}$ by constructing its quotient graph.
	\subsection{The construction of quotient graph}
	\par It is known that if a graph $\Gamma$ has a group $G$ action on it, and $\rho_s,s=0,1,2,\dots,n$ are all the irreducible representations of $G$, then the spectrum of the Laplacian $\mathscr{H}$ on $\Gamma$ satisfies
	$$
	\sigma(\mathscr{H})=\bigcup^{n}_{s=1}\sigma(\mathscr{H}_s).
	$$
	If $\rho_n$ is the $D$-dimensional irreducible representation of $G$, then the quotient graph $G/\rho_n$ of $\Gamma$ is obtained by gluing $D$ basic domains at their boundary points. The new vertex conditions formed during this gluing process depend on $\rho_n$. Moreover, the spectrum $\sigma$ of the quotient graph $G/\rho_n$ is isomorphic to $\sigma(\mathscr{H}_n)$. Since all irreducible representations of the cyclic group are one-dimensional, the quotient graph of $\Gamma_{n_1}\Box\Gamma_{n_2}$ consists of a single fundamental domain.
	\par Consider the Cartesian product graph $\Gamma_{n_1}\Box\Gamma_{n_2}$ with standard condition. It is evident that the group $G_{n_1}\otimes G_{n_2}$ is a cyclic groups\upcite{ref14}, and according to Theorem 4.1, it admits $n_1n_2$ quotient graphs. The vertex conditions for the quotient graph $\Gamma_{n_1}\Box\Gamma_{n_2}/\tau_{s,t}$ are shown as follows.
	\par Introducing virtual vertices at the center of edges in $\Gamma_{n_1}\Box\Gamma_{n_2}$. Let the fundamental domain of $\Gamma_{n_1}\Box\Gamma_{n_2}$ is a graph which vertices are $v_1,v_a,v_b,v_c,v_d$(Figure 5.1(a)). After the action of the group $G_{n_1}\otimes G_{n_2}=\left\{(g_1^\kappa,g_2^\iota)\right\},\kappa=1,\dots,n_1,\iota=1,\dots,n_2$, the vertices $v_a$ and $v_d$, $v_b$ and $v_c$ are glued together to form new vertices $\widetilde{v_d}$ and $\widetilde{v_c}$ respectively, thus creating the quotient graph of $\Gamma_{n_1}\Box\Gamma_{n_2}$ (some vertices are omitted and edge lengths are neglected in Figure 5.1 for clarity).
	
	\begin{figure}[H]
		\centering 
		\includegraphics[width=10cm]{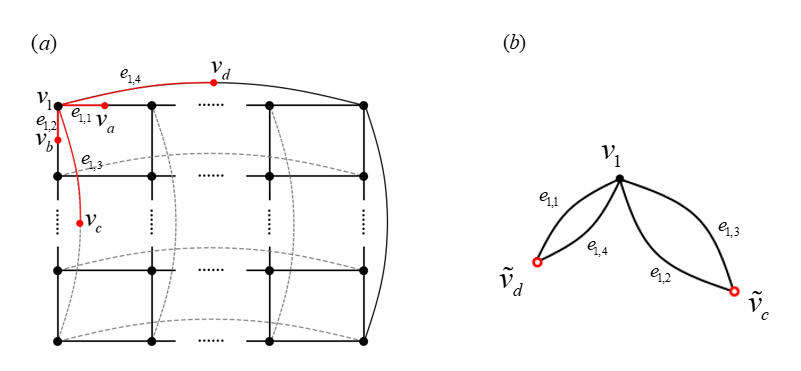}\\
		\caption{(a)The fundamental domain of $\Gamma_{n_1}\Box\Gamma_{n_2}$.(b)The quotient graph of $\Gamma_{n_1}\Box\Gamma_{n_2}$.}
		\label{fig:fig5}
	\end{figure}
	
	\par After the group $G_{n_1}\otimes G_{n_2}$ actions, we have:
	\begin{equation}
		\left\{
		\begin{aligned}
			&(g_1^{n_1},g_2)f\mid_{e_{1,1}}(v_a):=\tau_{s,t}(g_1^{n_1},g_2)f\mid_{e_{1,1}}(v_a)=f\mid_{(g_1^{n_1},g_2)e_{1,1}}(v_d),\\
			&(g_1^{n_1},g_2)f^\prime\mid_{e_{1,1}}(v_a):=\tau_{s,t}(g_1^{n_1},g_2)f^\prime\mid_{e_{1,1}}(v_a)=f^\prime\mid_{(g_1^{n_1},g_2)e_{1,1}}(v_d),
		\end{aligned} 
		\right. 
	\end{equation}
	\begin{equation}
		\left\{
		\begin{aligned}
			&(g_1,g_2^{n_2})f\mid_{e_{1,2}}(v_b):=\tau_{s,t}(g_1,g_2^{n_2})f\mid_{e_{1,2}}(v_b)=f\mid_{(g_1,g_2^{n_2})e_{1,2}}(v_c),\\
			&(g_1,g_2^{n_2})f^\prime\mid_{e_{1,2}}(v_b):=\tau_{s,t}(g_1,g_2^{n_2})f^\prime\mid_{e_{1,2}}(v_b)=f^\prime\mid_{(g_1,g_2^{n_2})e_{1,2}}(v_c),
		\end{aligned} 
		\right. 
	\end{equation}
	The standard condition satisfied at the vertices $\widetilde{v_d}$ and $\widetilde{v_c}$ are as follows: 
	\begin{equation}
		\left\{
		\begin{aligned}
			&f\mid_{e_{1,4}}(v_d)=f\mid_{(g_1^{n_1},g_2)e_{1,1}}(v_d),\\
			&f^\prime\mid_{e_{1,4}}(v_d)+f^\prime\mid_{(g_1^{n_1},g_2)e_{1,1}}(v_d)=0,
		\end{aligned}
		\right.
		\left\{
		\begin{aligned}
			&f\mid_{e_{1,3}}(v_c)=f\mid_{(g_1,g_2^{n_2})e_{1,2}}(v_c),\\
			&f^\prime\mid_{e_{1,3}}(v_c)+f^\prime\mid_{(g_1,g_2^{n_2})e_{1,2}}(v_c)=0,
		\end{aligned}
		\right.
	\end{equation}
	Therefore, from equations (5.1) to (5.3), we have:
	\begin{equation}
		\left\{
		\begin{aligned}
			&\tau_{s,t}(g_1^{n_1},g_2)f\mid_{e_{1,1}}(v_a)-f\mid_{e_{1,4}}(v_d)=0,\\
			&\tau_{s,t}(g_1^{n_1},g_2)f^\prime\mid_{e_{1,1}}(v_a)+f^\prime\mid_{e_{1,4}}(v_d)=0,
		\end{aligned}
		\right.
		\left\{
		\begin{aligned}
			&\tau_{s,t}(g_1,g_2^{n_2})f\mid_{e_{1,2}}(v_b)-f\mid_{e_{1,3}}(v_c)=0,\\
			&\tau_{s,t}(g_1,g_2^{n_2})f^\prime\mid_{e_{1,2}}(v_b)+f^\prime\mid_{e_{1,3}}(v_c)=0.
		\end{aligned}
		\right.
	\end{equation}
	Then
	\begin{equation}
		A_{\tilde{v}_d}=
		\begin{pmatrix}
			\tau_{s,t}(g_1^{n_1},g_2) & -1\\
			0           & 0 
		\end{pmatrix}
		,
		B_{\tilde{v}_d}=
		\begin{pmatrix}
			0           & 0\\
			\tau_{s,t}(g_1^{n_1},g_2) & 1 
		\end{pmatrix}
		,
	\end{equation}
	\begin{equation}
		A_{\tilde{v}_c}=
		\begin{pmatrix}
			\tau_{s,t}(g_1,g_2^{n_2}) & -1\\
			0           & 0 
		\end{pmatrix}
		,
		B_{\tilde{v}_c}=
		\begin{pmatrix}
			0           & 0\\
			\tau_{s,t}(g_1,g_2^{n_2}) & 1 
		\end{pmatrix}
		,
	\end{equation}
	\begin{equation}
		A_{v_1}=
		\begin{pmatrix}
			1 &-1 &0 &0\\
			0 &1  &-1 & 0\\
			0 &0  &1  &-1\\
			0 &0  &0  &0 
		\end{pmatrix}
		,
		B_{v_1}=
		\begin{pmatrix}
			0 &0  &0  &0\\
			0 &0  &0  &0\\
			0 &0  &0  &0\\
			1 &1  &1  &1 
		\end{pmatrix}
		.
	\end{equation}
	\par The function $f_{s,t}$ restricted to the quotient graph $\Gamma_{n_1}\Box\Gamma_{n_2}/\tau_{s,t}$ obviously still satisfies the corresponding vertex conditions. For convenience, we consider $v_i$ as the starting point for $i=1,2,\ldots,n_1n_2$. The edges $e_{i,1}$ and $e_{i,4}$ from $v_i$ to the virtual vertices are regarded as the interval $[0,L_1]$, and the edges $ e_{i,2}$ and $e_{i,3}$ from $v_i$ to the virtual vertices are regarded as the interval $[0,L_3]$. Acording to the proof of Theorem 4.1, there is a mapping $P_{s,t}$ from $f\in L^2(\Gamma_{n_1}\Box\Gamma_{n_2})$ to $f_{s,t}\in \mathcal{F}_{s,t}$, defined by Equation (4.1). Given that $f$ satisfies the standard conditions at the vertices, it can be shown that $f_{s,t}$ restricted to the quotient graph $\Gamma_{n_1}\Box\Gamma_{n_2}/\tau_{s,t}$ maintains the corresponding vertex conditions. 
	\paragraph{For the case $gcd(n_1,n_2)=1$.} We just need consider the circulant graph $C_{n_1n_2}(n_1,n_2)$ with standard condition now. It is evident that the circulant graph $C_{n_1n_2}(n_1,n_2)$ is a central symmetric graph with group $G_{n_1n_2}$ action. According to the decomposition of the function space Theorem 3.2 in \cite{ref17}, it has $n_1n_2$ quotient graphs. Similarly, we give the vertex conditions for quotient graph $C_{n_1n_2}(n_1,n_2)/\rho_\epsilon,\epsilon=0,1,2,\dots,n_1n_1-1$ (some vertices are omitted and edge lengths are neglected in Figure 5.2 for clarity).
	\begin{figure}[H]
		\centering 
		\includegraphics[width=12cm]{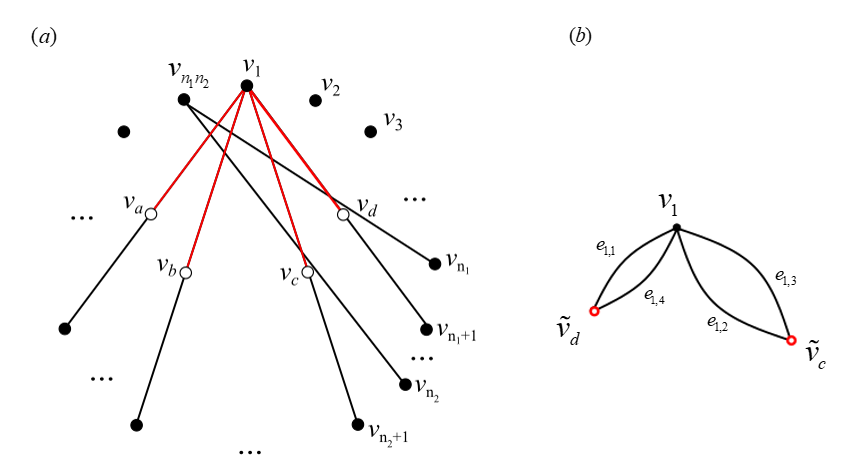}\\
		\caption{(a)Circulant graph $C_{n_1n_2}(n_1,n_2)$ and its domain $W$.(b)The quotient graph of $C_{n_1n_2}(n_1,n_2)$.}
		\label{fig:fig6}
	\end{figure}
	\par After the action of group $G_{n_1n_2}$ and with the standard condition at vertices $\widetilde{v_d}$ and $\widetilde{v_c}$, we have:
	\begin{equation}
		\left\{
		\begin{aligned}
			&\rho_\epsilon(g^{n_1})f\mid_{e_{1,1}}(v_a)-f\mid_{e_{1,4}}(v_d)=0,\\
			&\rho_\epsilon(g^{n_1})f^\prime\mid_{e_{1,1}}(v_a)+f^\prime\mid_{e_{1,4}}(v_d)=0,
		\end{aligned}
		\right.
		\left\{
		\begin{aligned}
			&\rho_\epsilon(g^{n_2})f\mid_{e_{1,2}}(v_b)-f\mid_{e_{1,3}}(v_c)=0,\\
			&\rho_\epsilon(g^{n_2})f^\prime\mid_{e_{1,2}}(v_b)+f^\prime\mid_{e_{1,3}}(v_c)=0.
		\end{aligned}
		\right.
	\end{equation}
	Then
	\begin{equation}
		A_{\tilde{v}_d}=
		\begin{pmatrix}
			\rho_\epsilon(g^{n_1}) & -1\\
			0           & 0 
		\end{pmatrix}
		,
		B_{\tilde{v}_d}=
		\begin{pmatrix}
			0           & 0\\
			\rho_\epsilon(g^{n_1}) & 1 
		\end{pmatrix}
		,
	\end{equation}
	\begin{equation}
		A_{\tilde{v}_c}=
		\begin{pmatrix}
			\rho_\epsilon(g^{n_2}) & -1\\
			0           & 0 
		\end{pmatrix}
		,
		B_{\tilde{v}_c}=
		\begin{pmatrix}
			0           & 0\\
			\rho_\epsilon(g^{n_2}) & 1 
		\end{pmatrix}
		,
	\end{equation}
	\begin{equation}
		A_{v_1}=
		\begin{pmatrix}
			1 &-1 &0 &0\\
			0 &1  &-1 & 0\\
			0 &0  &1  &-1\\
			0 &0  &0  &0 
		\end{pmatrix}
		,
		B_{v_1}=
		\begin{pmatrix}
			0 &0  &0  &0\\
			0 &0  &0  &0\\
			0 &0  &0  &0\\
			1 &1  &1  &1 
		\end{pmatrix}
		.
	\end{equation}
	\par The function $f_\epsilon$ restricted to the quotient graph $ C_{n_1n_2}(n_1,n_2)/\rho_\epsilon$ still satisfies the corresponding vertex conditions. For convenience in calculation, consider $v_i$ as the starting point for $i = 1,2,\dots,n_1n_2$. The edges $e_{i,1}$ and $e_{i,4}$ from  $v_i$ to the virtual vertices are regarded as the interval $[0,L_1]$, and the edges $e_{i,2}$ and $e_{i,3}$ from $v_i$ to the virtual vertices are regarded as the interval $[0,L_3]$. Similarly, according to the proof of Theorem 3.2 in \cite{ref17}, we have the mapping $P_\epsilon$ from $f\in L^2(C_{n_1n_2}(n_1,n_2))$to $f_\epsilon\in \mathcal{F}_\epsilon$ for $\epsilon = 0,1,\dots,n_1n_2-1$:
	$$
	P_\epsilon f|_{e_{i,m}} = f_\epsilon|_{e_{j,m}} = \frac{\rho_\epsilon(g^{n_1n_2})f|_{e_{i,m}} + \rho_\epsilon(g)f|_{ge_{i,m}} + \cdots + \rho_\epsilon(g^{n_1n_2-1})f|_{g^{n_1n_2-1}e_{i,m}}}{n_1n_2}.
	$$
	\par Given that $f$ satisfies the standard conditions at the vertices, it can be shown that $f_\epsilon$ restricted to the quotient graph $C_{n_1n_2}(n_1,n_2)/\rho_\epsilon$ still satisfies the corresponding vertex conditions.
	\begin{example}
		\rm Considering the quotient graph of~$C_{12}(3,4)$~with $G_{12}$ action. It has totally of 12 quotient graphs and for $C_{n_1n_2}(n_1,n_2)/\rho_\epsilon$ the vertex condition is:
		\begin{equation}
			A_{\tilde{v}_d}=
			\begin{pmatrix}
				\rho_\epsilon(g^3) & -1\\
				0           & 0 
			\end{pmatrix}
			,
			B_{\tilde{v}_d}=
			\begin{pmatrix}
				0           & 0\\
				\rho_\epsilon(g^3) & 1 
			\end{pmatrix}
			,
		\end{equation}
		\begin{equation}
			A_{\tilde{v}_c}=
			\begin{pmatrix}
				\rho_\epsilon(g^4) & -1\\
				0           & 0 
			\end{pmatrix}
			,
			B_{\tilde{v}_d}=
			\begin{pmatrix}
				0           & 0\\
				\rho_\epsilon(g^4) & 1 
			\end{pmatrix}
			,
		\end{equation}
		\begin{equation}
			A_{v_1}=
			\begin{pmatrix}
				1 &-1 &0 &0\\
				0 &1  &-1 & 0\\
				0 &0  &1  &-1\\
				0 &0  &0  &0 
			\end{pmatrix}
			,
			B_{v_1}=
			\begin{pmatrix}
				0 &0  &0  &0\\
				0 &0  &0  &0\\
				0 &0  &0  &0\\
				1 &1  &1  &1 
			\end{pmatrix}
			,
		\end{equation}
		and $f_\epsilon$ restricts to the quotient graph $G_{12}/\rho_\epsilon$ satisfying the corresponding vertex condition. 
	\end{example}
	\subsection{Decomposition of the secular determinant}
	\par The secular determinant is a function of the eigenvalues of differential operators on quantum graphs. In this chapter, the secular determinant of $L^2(\Gamma_{n_1}\Box\Gamma_{n_2})$ is obtained by the calculation method in reference \cite{ref13}. Additionally, the decomposition of the secular determinant is derived by applying the decomposition theorem of function space from Theorem 4.1. All factors of secular determinant of $\Gamma_{n_1}\Box\Gamma_{n_2}$ correspond to the secular determinants of all its quotient graphs. This allows the study of the eigenvalues of the original quantum graph to be transformed into the study of the eigenvalues of its quotient graphs. We just present a lemma concerning the scattering matrix:
	\begin{lemmar}\rm\upcite{ref13}
		\rm $\lambda=k^2\neq0$ is an eigenvalue of the Laplace operator on metric graphs if and only if $k$ is a root of the following equation,
		\begin{equation}
			\varSigma(k):=\det(I-SD(k))=0,
		\end{equation}
		and $S$ is bond scattering matrix, $D(k)$ is a diagonal matrix related to edge lengths, this equation is known as the \textit{secular equation}.
	\end{lemmar}
	\par In Section 5.1, we obtained the quotient graph of $\Gamma_{n_1}\Box\Gamma_{n_2}$ and its vertex conditions (5.5) to (5.7). In this section, we will provide the decomposition of the secular equation for the quotient graph with different vertex conditions. We assign directions and values to the four edges of its quotient graph firstly, with $L_1 = L_2$ and $L_3 = L_4$. When $k\neq0$, the solution to the equation $-f''=k^2f$ on each edge can be written as:
	\begin{equation}
		f_j(x)=a_je^{ikx}+a_{\overline{j}}e^{ik(L_j-x)},j=1,2,3,4,
	\end{equation}
	\begin{figure}[H]
		\centering
		\includegraphics[width=40mm]{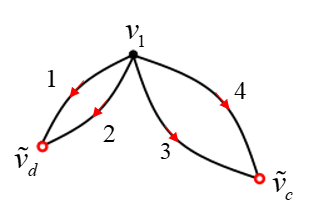}\\
		\caption{The quotient graph of $\Gamma_{n_1}\Box\Gamma_{n_2}$ and the directions of bonds.}
	\end{figure}
	The diagonal matrix $D(k)$ related to the edge lengths is given by:
	\begin{equation}
		D(k)=
		\begin{pmatrix}
			e^{ikL_1} &0 &0 &0 &0 &0 &0 &0\\
			0 &e^{ikL_2} &0 &0 &0 &0 &0 &0\\
			0 &0 &e^{ikL_3} &0 &0 &0 &0 &0\\
			0 &0 &0 &e^{ikL_4} &0 &0 &0 &0\\
			0 &0 &0 &0 &e^{ikL_1} &0 &0 &0\\
			0 &0 &0 &0 &0 &e^{ikL_2} &0 &0\\
			0 &0 &0 &0 &0 &0 &e^{ikL_3} &0\\
			0 &0 &0 &0 &0 &0 &0 &e^{ikL_4}\\
		\end{pmatrix}.
	\end{equation}
	According to (5.5)-(5.7):
	\begin{equation}
		\left\{
		\begin{aligned}
			&\tau_{s,t}(g_1^{n_1},g_2)f_1(L_1)-f_2(L_2)=0,\\
			&-\tau_{s,t}(g_1^{n_1},g_2)f'_1(L_1)-f'_2(L_2)=0,
		\end{aligned}
		\right.
		\left\{
		\begin{aligned}
			&\tau_{s,t}(g_1,g_2^{n_2})f_3(L_3)-f_4(L_4)=0,\\
			&-\tau_{s,t}(g_1,g_2^{n_2})f'_3(L_1)-f'_4(L_4)=0.
		\end{aligned}
		\right.
	\end{equation}
	Substituting (5.11) into (5.13):
	\begin{equation}
		\begin{aligned}
			&a_{\overline{1}}=\tau_{s,t}(g_1^{n_1},g_2)^{-1}a_2e^{ikL_2},
			a_{\overline{2}}=\tau_{s,t}(g_1^{n_1},g_2)a_1e^{ikL_1},\\
			&a_{\overline{3}}=\tau_{s,t}(g_1,g_2^{n_2})^{-1}a_4e^{ikL_4},
			a_{\overline{4}}=\tau_{s,t}(g_1,g_2^{n_2})a_3e^{ikL_3}.
		\end{aligned}
	\end{equation}
	The standard conditions at vertex $v_1$:
	\begin{equation}
		\begin{aligned}
			&f_1(0)=f_2(0)=f_3(0)=f_4(0);\\
			&f'_1(0)+f'_2(0)'+f'_3(0)+f_4(0)=0,
		\end{aligned}
	\end{equation}
	Sunstituting into (5.11)
	\begin{equation}
		a_i=-a_{\overline{i}}e^{ikL_i}+\frac{1}{2}\varSigma^4_{j=1}a_{\overline{j}}e^{ikL_j},j=1,2,3,4.
	\end{equation}
	Then the secular determinant is obtained as:
	\begin{equation}
		\nonumber
		S_{s,t}=
		\begin{pmatrix}
			0 &0 &0 &0 &-1/2 &1/2 &1/2 &1/2\\
			0 &0 &0 &0 &1/2 &-1/2 &1/2 &1/2\\
			0 &0 &0 &0 &1/2 &1/2 &-1/2 &1/2\\
			0 &0 &0 &0 &1/2 &1/2 &1/2 &-1/2\\
			0 &\tau_{s,t}(g_1^{n_1},g_2)^{-1} &0 &0 &0 &0 &0 &0\\
			\tau_{s,t}(g_1^{n_1},g_2) &1 &0 &0 &0 &0 &0 &0\\
			0 &0 &0  &\tau_{s,t}(g_1,g_2^{n_2})^{-1} &0 &0 &0 &0\\
			0 &0 &\tau_{s,t}(g_1,g_2^{n_2}) &0 &0 &0 &0 &0\\
		\end{pmatrix}.
	\end{equation}
	Therefore, by Lemma 5.2, the secular determinant for the quotient graph $\Gamma_{n_1}\Box\Gamma_{n_2}/\tau_{s,t}$ of $\Gamma_{n_1}\Box\Gamma_{n_2}$ is given by:
	\begin{equation}
		\begin{aligned}
			\varSigma_{s,t}(k)=&1-\frac{1}{2}(\tau_{s,t}(g_1^{n_1},g_2)+\tau_{s,t}(g_1^{n_1},g_2)^{-1})e^{2ikL_1}-\frac{1}{2}(\tau_{s,t}(g_1,g_2^{n_2})+\tau_{s,t}(g_1,g_2^{n_2})^{-1})e^{2ikL_3}\\
			&+\frac{1}{2}(\tau_{s,t}(g_1^{n_1},g_2)+\tau_{s,t}(g_1^{n_1},g_2)^{-1})e^{ik(2L_1+4L_3)}+\frac{1}{2}(\tau_{s,t}(g_1,g_2^{n_2})+\tau_{s,t}(g_1,g_2^{n_2})^{-1})e^{ik(4L_1+2L_3)}\\
			&-e^{ik4(L_1+L_3)}.
		\end{aligned}
	\end{equation}
	Thus, the secular determinant of $\Gamma_{n_1}\Box\Gamma_{n_2}$ can be decomposed as:
	\begin{equation}
		\varSigma_{\Gamma_{n_1}\Box\Gamma_{n_2}}(k)=\prod^{n_1-1}_{s=0}\prod^{n_2-1}_{t=0}\varSigma_{s,t}(k).
	\end{equation}
	\par Furthermore, the choice of direction for the edges in quotient graph does not affect the result of the secular determinant. For instance, if we change the direction of edges 2 and 4, the diagonal matrix $D(k)$ remains unchanged, while the scattering matrix for the $t$-th quotient graph becomes:
	\begin{equation}
		\nonumber
		S'_{s,t}=
		\begin{pmatrix}
			0 &1/2 &0 &1/2 &-1/2 &0 &1/2 &0\\
			\tau_{s,t}(g_1^{n_1},g_2) &0 &0 &0 &0 &0 &0 &0\\
			0 &1/2 &0 &1/2 &1/2 &0 &-1/2 &0\\
			0 &0 &\tau_{s,t}(g_1,g_2^{n_2}) &0 &0 &0 &0 &0\\
			0 &0 &0 &0 &0 &\tau_{s,t}(g_1^{n_1},g_2)^{-1} &0 &0\\
			0 &-1/2 &0 &1/2 &1/2 &0 &1/2 &0\\
			0 &0 &0  &0 &0 &0 &0 &\tau_{s,t}(g_1,g_2^{n_2})\\
			0 &1/2 &0 &-1/2 &1/2 &0 &1/2 &0\\
		\end{pmatrix}.
	\end{equation}
	After calculation, the secular determinant \(\varSigma'_{s,t}(k) = \varSigma_{s,t}(k)\). Therefore, the choice of direction for the bonds does not affect the result of the secular equation.
	\par For the case $gcd(n_1,n_2)=1$, expressing $\tau_{s,t}(g_1^{n_1},g_2),\tau_{s,t}(g_1,g_2^{n_2})$ in terms of $\rho_{\epsilon}(g^{n_1}), \rho_{\epsilon}(g^{n_2}), \epsilon=(sn_2+tn_1)\mod (n_1n_2)$ from Remark 4.2.
	\begin{example}
		\rm The secular determinant of the qutient graph $G_{12}/\rho_\epsilon$ of $C_{12}(3,4)$ is 
		$$
		\begin{aligned}
			\varSigma_\epsilon(k)=&1-\frac{1}{2}(\rho_\epsilon(a^{3})+\rho_\epsilon(a^{3})^{-1})e^{2ikL_1}-\frac{1}{2}(\rho_\epsilon(a^{4})+\rho_l(a^{4})^{-1})e^{2ikL_3}\\
			&+\frac{1}{2}(\rho_\epsilon(a^{3})+\rho_\epsilon(a^{3})^{-1})e^{ik(2L_1+4L_3)}+\frac{1}{2}(\rho_\epsilon(a^{4})+\rho_\epsilon(a^{4})^{-1})e^{ik(4L_1+2L_3)}\\
			&-e^{ik4(L_1+L_3)}.
		\end{aligned}
		$$
		Therefore, the secular determinant of $\Gamma_{3}\Box\Gamma_{4}$ can be decomposed into:
		$$
		\varSigma_{\Gamma_{3}\Box\Gamma_{4}}=\varSigma_{C_{12}(3,4)}(k)=\prod^{11}_{\epsilon=0}\varSigma_\epsilon(k).
		$$
	\end{example}
	
\end{document}